\documentclass[11pt]{article}

\usepackage{graphicx}
\usepackage[margin=1in]{geometry}
\usepackage{amsmath}
\usepackage{amssymb}
\usepackage{amsthm}
\usepackage{mdframed}
\usepackage{stmaryrd}
\usepackage{mathtools}
\usepackage{enumitem}
\usepackage{placeins}
\usepackage{url}
\usepackage{comment}
\usepackage{tikz}

\newtheorem{theorem}{Theorem}
\newtheorem{obs}[theorem]{Observation}
\newtheorem{lem}[theorem]{Lemma}
\newtheorem{cor}[theorem]{Corollary}
\newtheorem*{notation}{Notation}{\bf}{\rm}
\newtheorem*{defn}{Definition}{\bf}{\rm}

\DeclareMathOperator{\dimH}{dim_H}
\DeclareMathOperator{\dimalg}{dim_{alg}}

\DeclareMathOperator{\TIME}{TIME}
\DeclareMathOperator{\SPACE}{SPACE}
\newcommand{\dimdelta}{\dim_{\Delta}}
\renewcommand{\dim}{{\mathrm{dim}}}

\newcommand{\R}{\mathbb{R}}

\newcommand{\N}{\mathbb{N}}
\newcommand{\Q}{\mathbb{Q}}
\newcommand{\NP}{\textup{NP}}

\newcommand{\E}{\textup{E}}
\newcommand{\EXP}{\textup{EXP}}
\newcommand{\DEC}{\textup{DEC}}
\newcommand{\ESPACE}{\textup{ESPACE}}
\newcommand{\EXPSPACE}{\textup{EXPSPACE}}

\newcommand{\disj}{\textup{disj}}
\newcommand{\disjNP}{\textup{disjNP}}
\newcommand{\disjEXP}{\textup{disjEXP}}
\newcommand{\pSEL}{\textup{p-SEL}}
\newcommand{\qpSEL}{\textup{qp-SEL}}
\newcommand{\up}{\textup{p}}
\newcommand{\uP}{\textup{P}}
\newcommand{\qp}{\textup{qp}}
\newcommand{\all}{\textup{all}}
\newcommand{\comp}{\textup{comp}}
\newcommand{\alg}{\textup{alg}}
\newcommand{\uspace}{\textup{space}}

\newcommand{\trialph}{$\{0, 1, -1\}$}

\newcommand{\bin}{$\{0,1\}$}
\newcommand{\regSS}{S^\infty}

\numberwithin{equation}{section}
\numberwithin{theorem}{section}

\begin{document}

\title{Dimension and the Structure of Complexity Classes}

\author{Jack H. Lutz\footnote{Research supported in part by National Science Foundation grants 1545028 and 1900716.}\\Iowa State University\and Neil Lutz\\Swarthmore College\and Elvira Mayordomo\footnote{Research supported in part by Spanish Ministry of Science, Innovation and Universities grant PID2019-104358RB-I00.}\\Universidad de Zaragoza}

\date{\vspace{1.5em}Dedicated to the memory of Alan L. Selman}

\maketitle

\begin{abstract}
	We prove three results on the dimension structure of complexity classes.
	\begin{enumerate}
		\item The Point-to-Set Principle, which has recently been used to prove several new theorems in fractal geometry, has resource-bounded instances. These instances characterize the resource-bounded dimension of a set $X$ of languages in terms of the relativized resource-bounded dimensions of the individual elements of $X$, provided that the former resource bound is large enough to parameterize the latter. Thus for example, the dimension of a class $X$ of languages in $\EXP$ is characterized in terms of the relativized $\up$-dimensions of the individual elements of $X$.

		\item Every language that is $\leq^P_m$-reducible to a $\up$-selective set has $\up$-dimension 0, and this fact holds relative to arbitrary oracles. Combined with a resource-bounded instance of the Point-to-Set Principle, this implies that if $\NP$ has positive dimension in $\EXP$, then no quasipolynomial time selective language is $\leq^P_m$-hard for $\NP$.

		\item If the set of all disjoint pairs of $\NP$ languages has dimension 1 in the set of all disjoint pairs of $\EXP$ languages, then $\NP$ has positive dimension in $\EXP$.
	\end{enumerate}
\end{abstract}

\section{Introduction}

Alan Selman was a pioneer and a leader in elucidating the structure of complexity classes. He initiated many of the most important concepts of structural complexity theory, he investigated them brilliantly, and he inspired generations of computer scientists to contribute to this endeavor.

Our objective in this paper is to show how resource-bounded dimension, which is a generalization of classical Hausdorff dimension, can extend Selman's research program in fruitful new directions. To this end, we present three new results, one bringing the Point-to-Set Principle into complexity classes, one on dimension and $\up$-selective sets, and one on dimension and disjoint $\NP$ pairs. The rest of this introduction motivates and explains these three results.

Hausdorff dimension, developed in 1919~\cite{Haus19,Falc14}, is a scheme for assigning a \emph{dimension} $\dimH(E)$ to every subset $E$ of a given metric space. Assume for a moment that this metric space is a Euclidean space $\R^n$. Then $\dimH(\R^n) = n$, and the Hausdorff dimension is monotone, i.e., $E \subseteq F$ implies that $\dimH(E) \leq \dimH(F)$. For integers $d = 0,\ldots,n$, subsets $E$ of $\R^n$ that are intuitively $d$-dimensional have $\dimH(E) = d$. However, \emph{every} real number $s \in [0,n]$ is the Hausdorff dimension of infinitely many (in fact, $2^{|\R|}$ many) subsets of $\R$. In general, $\dimH(E) < n$ implies that $E$ is a Lebesgue measure 0 subset of $\R^n$. (The converse does not hold.) Hausdorff dimension can thus be regarded as a measure of the ``sizes'' of Lebesgue measure 0 subsets of $\R^n$. Hausdorff dimension has become a powerful tool for investigations in fractal geometry, probability theory, and other areas of mathematical analysis~\cite{Falc14,SteSha05,Matt15,BisPer17}.

We momentarily shift the focus of our discussion from Euclidean spaces $\R^n$ to another metric space, the \emph{Cantor space} $\mathbf{C}$ consisting of all \emph{decision problems}, which are equivalently regarded as subsets of $\{0,1\}^*$ or as infinite binary sequences. At the beginning of the present century, the first author proved a theorem characterizing Hausdorff dimension in $\mathbf{C}$ in terms betting strategies called \emph{gales}, which are minor but convenient generalization of martingales. Based on this characterization, he introduced two related methods for \emph{effectivizing} Hausdorff dimension, i.e., imposing computability or complexity constraints on these gales. The first of these methods~\cite{Lutz03a}, called \emph{resource-bounded dimension} imposes Hausdorff dimension structure on complexity classes. For example this theory defines, for every subset $X$ of $\mathbf{C}$, a quasipolynomial-time (i.e., $n^{\textup{polylog}\,n}$-time) dimension $\dim_\qp(X)$ in such a way that $\dim(X\mid \EXP) = \dim_\qp(X \cap \EXP)$ is a coherent notion of the dimension of $X$ within the complexity class $\EXP = \TIME(2^\textup{polynomial})$. The second method~\cite{DISS}, \emph{algorithmic dimension} (also called \emph{constructive dimension} or \emph{effective dimension}) has to date been more widely investigated, partly because of its interactions with algorithmic randomness (i.e., Martin-L\"{o}f randomness~\cite{MarL66}) and partly because of its applications to classical fractal geometry~\cite{LutLut20,LutMay21}. Algorithmic dimension plays a motivating role in this paper, but resource-bounded dimension is our main topic.

Several recent results in algorithmic fractal dimensions are based on the 2017 Point-to-Set Principle introduced by the first two authors~\cite{LutLut17}. This principle is a family of theorems, the first of which says that, for any set $E \subseteq \R^n$,
\begin{equation}\label{eq:psp}
  \dimH(E) = \adjustlimits\min_{A\in\mathbf{C}}\sup_{x\in E}\dim^A(x),
\end{equation}
where $\dim^A(x)$ is the algorithmic dimension of the individual point $x$ relative to the oracle $A$. This theorem completely characterizes the \emph{classical} Hausdorff dimensions of sets $E$ in terms of the relativized algorithmic dimensions of their elements $x$. The term ``classical'' here does not mean ``old,'' but rather refers to mathematical concepts and theorems that, like Hausdorff dimension, do not involve computability or logic in their formulations. Thus the left-hand side of~\eqref{eq:psp} is classical, but the right-hand side, involving computability, is not. The characterization theorem~\eqref{eq:psp} is called the \emph{Point-to-Set Principle for Hausdorff dimension}, because it enables one to prove lower bounds on the Hausdorff dimensions of \emph{sets} by reasoning about the relativized algorithmic dimensions of judiciously chosen individual \emph{points} in those sets. The paper~\cite{LutLut17} also proved a second instance of the Point-to-Set Principle that characterizes another classical fractal dimension, the packing dimension~\cite{Falc14}, in a manner dual to~\eqref{eq:psp}. These instances of the Point-to-Set Principle have recently been used to prove several new theorems in classical fractal geometry~\cite{LutStu20,LutStu18,Lutz21,LuLuMa20}. The authors also recently extended \eqref{eq:psp} and its dual from $\R^n$ to arbitrary separable metric spaces and to Hausdorff and packing dimensions with very general gauge families~\cite{LuLuMa20}.

The above instances of the Point-to-Set Principle characterize classical fractal dimensions of sets in terms of the relativized algorithmic dimensions of the individual elements of those sets. In Section~\ref{sec:psp} below, we prove more general instances of the Point-to-Set Principle that characterize the \emph{classical or perhaps somewhat effective} dimensions of sets in $\mathbf{C}$ in terms of the relativized \emph{more effective} dimensions of the individual elements of those sets. One example of this says that, for every subset $X$ of $\mathbf{C}$,
\begin{equation}\label{eq:intro2}
  \dimH(X)= \adjustlimits\min_{B\in\mathbf{C}}\sup_{A\in X}\dim_\up^B(A).
\end{equation}
That is, we can replace the algorithmic dimension on the right-hand side of~\eqref{eq:psp} by the more effective polynomial-time dimension. Another example characterizes the quasipolynomial-time dimension of each subset $X$ of $\mathbf{C}$ by
\begin{equation}\label{eq:intro3}
  \dim_\qp(X)= \adjustlimits\min_{B\in \EXP}\sup_{A\in X}\dim_\up^{\langle B\rangle}(A),
\end{equation}
i.e., in terms of the more effective polynomial-time dimensions of the individual elements $A$ of $X$. (The ``$\langle B\rangle$'' refers to a technically restricted relativization of $\up$-dimension to the oracle $B$ explained in Section~\ref{sec:psp}.). This implies that, for every subset $X$ of $\mathbf{C}$ and every $\EXP$-complete language $C$,
\begin{equation}\label{eq:intro4}
  \dim(X\mid \EXP) = \sup_{A\in X\cap\EXP}\dim_\up^{\langle C\rangle}(A).
\end{equation}
The instances~\eqref{eq:intro2},~\eqref{eq:intro3}, and~\eqref{eq:intro4} are all special cases of Theorem~\ref{thm:psprbd} in Section~\ref{sec:psp}.

In 1979, Alan Selman adapted Jockusch's computability-theoretic notion of semirecursive sets \cite{Jock68}, creating the complexity-theoretic notion of $\up$-selective sets~\cite{Selm79}. Briefly, a decision problem $A \subseteq \{0,1\}^*$ is \emph{$\up$-selective}, and we write $A \in \pSEL$, if there is a polynomial-time algorithm that, given an ordered pair $(x,y)$ of strings $x,y \in \{0,1\}^*$, outputs a string $z \in \{x,y\}$ such that $\{x,y\} \setminus A\neq \emptyset \implies z \in A$. (We note that the terms``$\up$-selective'' and ``$\uP$-selective'' have both been widely used for this notion. In fact, both have been used in papers with Selman as an author.) Every set $A \in P$ is clearly $\up$-selective, but there are uncountably many $\up$-selective sets, so the converse does not hold. There is an extensive literature on $\up$-selective sets and the related notions that they have spawned. We especially refer the reader to the books by Hemaspaandra and Torenvliet~\cite{HemTor03} and Zimand~\cite{Zima04} and the references therein.

Selman~\cite{Selm79} proved that no $\up$-selective set can be $\leq^P_m$-hard for $\EXP$ and that, if $\uP\neq\NP$, then no $\up$-selective set can be $\leq^P_m$-hard for $\NP$. In order to extend the class of provably intractable problems, the first author~\cite{Lutz90} defined a language $H$ to be \emph{weakly $\leq^P_m$-hard} for $\EXP$ if $\mu(P_m(H)\mid\EXP) \neq 0$, i.e., if the set $P_m(H)$ of languages $A$ such that $A\leq^P_m H$ does not have measure 0 in $\EXP$ in the sense of resource-bounded measure~\cite{AEHNC,QSET,Zima04}. Buhrman and Longpr\'{e}~\cite{BuhLon96} and, independently, Wang~\cite{Wang96} proved that $\mu(P_m(\pSEL)\mid\EXP) = 0$, where for a class $X\subseteq\mathbf{C}$, $P_m(X)=\bigcup_{H\in X}(P_m(H))$. It follows that no $\up$-selective set can be weakly $\leq^P_m$-hard for $\EXP$. (They in fact proved the stronger fact that this also holds for $\leq^P_{tt}$-reductions.) See~\cite{Zima04} for a host of related results.

After the development of resource-bounded dimension~\cite{Lutz03a}, Ambos-Spies, Merkle, Reimann, and Stephan~\cite{AMRS01} defined a language $H$ to be \emph{partially $\leq^P_m$-hard} for $\EXP$ if $\dim(P_m(H)\mid \EXP) > 0$. It is clear that weak hardness implies partial hardness, and it was shown in~\cite{AMRS01} that the converse does not hold. In Section~\ref{sec:sel} we use Theorem~\ref{thm:psprbd} (i.e., the Point-to-Set Principle) to prove that $\dim(P_m(\qpSEL)\mid \EXP) = 0$, where the set $\qpSEL$ of $\qp$-selective sets is the obvious quasipolynomial-time analog of $\pSEL$. This implies that no $\qp$-selective set can be partially $\leq^P_m$-hard for $\EXP$ and that, if $\dim(\NP\mid\EXP) > 0$, then no $\qp$-selective set can be $\leq^P_m$-hard for $\NP$.

In 1984, Even, Selman, and Yacobi~\cite{EvSeYa84} defined a \emph{promise problem} to be an ordered pair $(A,B)$ of disjoint languages. A solution of a promise problem $(A,B)$ is an algorithm or other device that decides \emph{any separator} of $(A,B)$, i.e., any language $S$ such that $A \subseteq S$ and $S \cap B = \emptyset$. Intuitively, we are \emph{promised} that every input will be an element of $A \cup B$, so we are only required to correctly distinguish inputs in $A$ from inputs in $B$.

A \emph{disjoint $\NP$ pair} is a promise problem $(A,B)$ with $A, B \in \NP$. Disjoint $\NP$ pairs were first investigated by Selman and collaborators to better understand public key cryptosystems~\cite{EvSeYa84,GroSel88,Selm89,HomSel92}. Razborov~\cite{Razb94} later established a deep connection between disjoint $\NP$ pairs and propositional proof systems, associating with each propositional proof system a canonical disjoint $\NP$ pair. Gla{\ss}er, Selman, Sengupta, and Zhang~\cite{GSSZ04,GlSeSe05,GlSeZh07,GlSeZh09} investigated this connection further, and it is now known that the degree structure of propositional proof systems under the natural notion of proof simulation is identical to the degree structure of disjoint $\NP$ pairs under reducibility of separators. See~\cite{GHSW14} for a survey of this and related results and~\cite{DosGla20} for more recent work.

In 2012, Fortnow, the first author, and the third author~\cite{FoLuMa12} investigated strong hypotheses involving the intractability of disjoint $\NP$ pairs. Among other things, this paper proved that
\begin{equation}\label{eq:intro5}
       \mu(\disjNP\mid\disjEXP) \neq 0 \implies \mu(\NP\mid\EXP) \neq 0
\end{equation}
and that $\mu(\NP\mid\EXP)\neq 0$ implies the existence, for every $k$, of disjoint $\NP$ pairs that cannot be separated in $2^{n^k}$ time. (Here $\disjNP$ is the set of disjoint $\NP$ pairs, and $\disjEXP$ is the set of disjoint $\EXP$ pairs, the latter endowed with a natural measure.)

In Section~\ref{sec:disj}, we prove a dimension-theoretic analog of~\eqref{eq:intro5}, namely that
\begin{equation}\label{eq:intro6}
    \dim(\disjNP\mid\disjEXP) = 1 \implies \dim(\NP\mid\EXP) > 0.
\end{equation}

Our proof of~\eqref{eq:intro6} is somewhat simplified by the use of Theorem~\ref{thm:psprbd} (i.e., the Point-to-Set Principle).

\section{Resource Bounds}\label{sec:rb}
We work in the Cantor space $\mathbf{C}$ consisting of all \emph{decision problems} (i.e., \emph{languages}) $A\subseteq\{0,1\}^*$. We identify each decision problem $A$ with its characteristic sequence
\[\llbracket s_0\in A\rrbracket\ \llbracket s_1\in A\rrbracket\ \llbracket s_2\in A\rrbracket\ldots,\]
where $s_0,s_1,s_2,\ldots$ is the standard enumeration of $\{0,1\}^*$ and 
\[\llbracket\varphi\rrbracket=\textbf{if}\ \varphi\ \textbf{then}\ 1\ \textbf{else}\ 0\]
is the Boolean value of a statement $\varphi$. We thus regard $\mathbf{C}$ as either the power set $\mathcal{P}(\{0,1\}^*)$ of $\{0,1\}^*$ or as the set $\{0,1\}^\omega$ of all infinite binary sequences, whichever is most convenient in a given context.

A \emph{resource bound} in this paper is any one of several classes of functions from $\{0,1\}^*$ to $\{0,1\}^*$ that we now specify.

The largest resource bound is the set
\[\all=\big\{f\mid f:\{0,1\}^*\to\{0,1\}^*\big\}\]
we also use the resource bound
\[\comp=\{f\in\all\mid f\ \text{is computable}\}.\]

As in~\cite{LutzThesis,Lutz92,Lutz03a}, we define a hierarchy $G_0,G_1,G_2,\ldots$ of classes of \emph{growth rates} $f:\N\to\N$ by the following recursion. (All logarithms in this paper are base-2.)
\begin{align*}
	G_0&=\{f\mid (\exists k)(\forall^\infty n)f(n)\leq kn\}\\
	G_{i+1}&=2^{G_i(\log n)}=\left\{f\;\middle|\;(\exists g\in G_i)(\forall^\infty n)f(n)\leq 2^{g(\log n)}\right\}.
\end{align*}
Note that $G_0$ is the class of $O(n)$ growth rates and that $G_1$ is the class of polynomially bounded growth rates. For each $i\in\N$, define a canonical growth rate $\hat{g}_i\in G_i$ by $\hat{g}_0(n)=2n$ and $\hat{g}_{i+1}(n)=2^{\hat{g}_i(\log n)}$. It is easy to verify that each $G_i$ is closed under composition, that each $f\in G_i$ is $o(\hat{g}_{i+1})$, and that each $\hat{g}_i$ is $o(2^n)$. Thus all growth rates in the $G_i$-hierarchy are subexponential.

Within the resource bound $\comp$, we use the resource bounds
\[\up_i=\{f\in\all\mid f\ \text{is computable in}\ G_i\ \text{time}\}\quad(i\geq 1)\]
and
\[\up_i\uspace=\{f\in\all\mid f\ \text{is computable in}\ G_i\ \text{space}\}\quad(i\geq 1).\]
(The length of the output \emph{is} included as part of the space used in computing $f$.) We write $\up$ for the polynomial-time resource bound $\up_1$ and $\qp$ for the quasipolynomial-time resource bound $\up_2$. Similarly the notations $\up\uspace$ and $\qp\uspace$ denote the space resource bounds $\up_1\uspace$ and $\up_2\uspace$, respectively.

In this paper, a \emph{resource bound} $\Gamma$ or $\Delta$ is one of the classes $\all$, $\comp$, $\up_i$ $(i\geq 1)$, $\up_i\uspace$ $(i\geq 1)$ defined above. We will also use relativizations $\Delta^A$ or $\Delta^g$ of a resource bound $\Delta$ to oracles $A\subseteq\{0,1\}^*$ or function oracles $g:\{0,1\}^*\to\{0,1\}^*$.

A \emph{constructor} is a function $\delta:\{0,1\}^*\to \{0,1\}^*$ such that $\delta(w)$ is a proper extension of $w$ (i.e., $w$ is a proper prefix of $\delta(w)$) for all $w\in\{0,1\}^*$. The \emph{result} of a constructor $\delta$ is the unique sequence $R(\delta)\in\mathbf{C}$ such that $\delta^n(\lambda)$ is a prefix of $R(\delta)$ for all $n\in\N$. (Here $\delta^n(\lambda)$ is the $n$-fold application of $\delta$ to the empty string $\lambda$.)

The \emph{result class} of a resource bound $\Delta$ is the class $R(\Delta)$ consisting of all languages $R(\delta)$ such that $\delta\in\Delta$ is a constructor. The following facts are easily verified.
\begin{enumerate}
	\item $R(\all)=\mathbf{C}$.
	\item $R(\comp)=\DEC$, the set of all decidable languages.
	\item For all $i\geq 1$,
	\[R(\up_i)=\E_i=\TIME(2^{G_{i-1}}).\]
	In particular,
	\[R(\up)=\E=\TIME(2^{\textup{linear}})\]
	and
	\[R(\qp)=\EXP=\TIME(2^{\textup{poly}}).\]
	\item For all $i\geq 1$,
	\[R(\up_i\uspace)=\E_i\SPACE=\SPACE(2^{G_{i-1}}).\]
	In particular,
	\[R(\up\uspace)=\ESPACE=\SPACE(2^{\textup{linear}})\]
	and
	\[R(\qp\uspace)=\EXPSPACE=\SPACE(2^{\textup{poly}}).\]
\end{enumerate}
Many of our functions will be of the form $f:D\to[0,\infty)$, where $D$ is a discrete domain such as $\{0,1\}^*$ or $\N\times\{0,1\}^*$ and $[0,\infty)$ is the set of nonnegative real numbers. If $\Delta$ is a resource bound, then such a function $f$ is \emph{$\Delta$-computable} if there is a rational-valued function $\hat{f}:D\times\N\to\Q\cap[0,\infty)$ such that $\lvert \hat{f}(r,x)-f(x)\rvert\leq 2^{-r}$ for all $x\in D$ and $r\in\N$ and $\hat{f}\in \Delta$ (with $r$ coded in unary and $\hat{f}(x,r)$ coded in binary).

We say that $f$ is \emph{lower semicomputable} if there is a computable function $\hat{f}:D\times \N\to\Q\cap[0,\infty)$ such that the following two conditions hold for all $x\in D$.
\begin{enumerate}
	\item[(i)] For all $t\in \N$, $\hat{f}(x,t)\leq\hat{f}(x,t+1)\leq f(x)$.
	\item[(ii)] $\displaystyle \lim_{t\to\infty}\hat{f}(x,t)=f(x)$.
\end{enumerate}

\section{Resource-Bounded Dimensions}

This section briefly reviews the elements of resource-bounded dimension developed in~\cite{Lutz03a}.

\begin{defn}
	\begin{enumerate}
		\item For $s\in[0,\infty)$, an \emph{$s$-gale} is a function $d:\{0,1\}^*\to[0,\infty)$ such that, for all $w\in\{0,1\}^*$,
		\[d(w)=2^{-s}[d(w0)+d(w1)].\]
		\item A \emph{martingale} is a 1-gale.
	\end{enumerate}
\end{defn}

\begin{obs}[\cite{DISS}]\label{obs:diss}
	A function $d:\{0,1\}^*\to[0,\infty)$ is an $s$-gale if and only if the function $d':\{0,1\}^*\to[0,\infty)$ defined by $d'(w)=2^{(1-s)|w|}d(w)$ is a martingale.
\end{obs}

An $s$-gale $d$ \emph{succeeds} on a language $A\subseteq\{0,1\}^*$, and we write $A\in S^\infty[d]$, if
\[\limsup_{w\to A} d(w)=\infty,\]
where the limit superior is taken over successively longer prefixes of $A$.

\begin{notation}
	For $X\subseteq\mathbf{C}$, let $\mathcal{G}(X)$ be the set of all $s\in[0,\infty)$ such that there is an $s$-gale $d$ for which $X\subseteq S^\infty[d]$.
\end{notation}

Readers unfamiliar with fractal geometry can safely use the following characterization as the definition of the \emph{Hausdorff dimension} $\dimH(X)$ of each set $X\subseteq\mathbf{C}$.

\begin{theorem}[gale characterization of Hausdorff dimension~\cite{Lutz03a}]\label{thm:galechar}
	For all $X\subseteq\mathbf{C}$,
	\[\dimH(X)=\inf\mathcal{G}(X).\]
\end{theorem}

Intuitively, an $s$-gale is a strategy for betting on the successive bits of languages $A\in\mathbf{C}$. The payoffs of these bets are fair if $s=1$ and unfair if $s<1$. Intuitively and roughly, Theorem~\ref{thm:galechar} says that the Hausdorff dimension of $X$ is the \emph{most hostile betting environment} in which a gambler can succeed on every language $A\in X$.

Motivated by the above characterization of classical Hausdorff dimension, the first author defined resource-bounded dimensions and algorithmic dimensions as follows.

\begin{notation}[\cite{Lutz03a,DISS}]
	Let $\Delta$ be a resource bound, and let $X\subseteq\mathbf{C}$.
	\begin{enumerate}
		\item $\mathcal{G}_\Delta(X)$ is the set of all $s\in[0,\infty)$ such that there is a $\Delta$-computable $s$-gale $d$ for which $X\subseteq S^\infty[d]$.
		\item $\mathcal{G}_{\alg}(X)$ is the set of all $s\in[0,\infty)$ such that there is a lower semicomputable $s$-gale $d$ for which $X\subseteq S^\infty[d]$.
	\end{enumerate}
\end{notation}

\begin{defn}[\cite{Lutz03a,DISS}]
	Let $\Delta$ be a resource bound, let $X\subseteq \mathbf{C}$, and let $A\in\mathbf{C}$.
	\begin{enumerate}
		\item\label{defn:dim1} The \emph{$\Delta$-dimension} of $X$ is
		\[\dim_\Delta(X)=\inf\mathcal{G}_\Delta(X).\]
		\item\label{defn:dim2} The \emph{$\Delta$-dimension} of $X$ \emph{in $R(\Delta)$} is
		\[\dim(X\mid R(\Delta))=\dim_\Delta(X\cap R(\Delta)).\]
		\item The \emph{$\Delta$-dimension} of $A$ is
		\[\dim_\Delta(A)=\dim_\Delta(\{A\}).\]
		\item\label{defn:dim4} The \emph{algorithmic dimension} of $X$ is
		\[\dimalg(X)=\inf\mathcal{G}_\alg(X).\]
		\item The \emph{algorithmic dimension} of $A$ is
		\[\dim(A)=\dimalg(\{A\}).\]
	\end{enumerate}
	(Algorithmic dimension has also been called \emph{constructive dimension} and \emph{effective dimension}.)
\end{defn}

The papers~\cite{Lutz03a,DISS} showed that the above-defined dimensions are coherent, well-behaved ``versions'' of Hausdorff dimension. All the defined dimensions lie in $[0,1]$, and all can take any real value in $[0,1]$. The dimensions~\ref{defn:dim1}.,~\ref{defn:dim2}., and~\ref{defn:dim4}., have the crucial dimension properties that they are monotone in $X$ and that they are \emph{stable} in the sense that the dimension of $X\cup Y$ is the maximum of the dimensions of $X$ and $Y$. Classical Hausdorff dimension (i.e., $\dimH=\dim_\all$) is also \emph{countably stable}, meaning that
\begin{equation}\label{eq:countablestability}
	\dimH\left(\bigcup_{i\in I} X_i\right)=\sup_{i\in I}\dimH(X_i)
\end{equation}
holds for all countable index sets $I$. The dimensions~\ref{defn:dim1}. and~\ref{defn:dim2}. are not countably stable for $\Delta$ smaller than $\text{all}$, but they are \emph{$\Delta$-countably stable} in that~\eqref{eq:countablestability} holds if the countable union is ``$\Delta$-effective.'' The algorithmic dimension~\ref{defn:dim4}. is \emph{absolutely stable} in the sense that~\eqref{eq:countablestability} holds, regardless of whether $I$ is countable. In particular, this implies that, for all $X\subseteq\mathbf{C}$,
\begin{equation}\label{eq:supdim}
	\dimalg(X)=\sup_{A\in X}\dim(A).
\end{equation}
As a consequence of~\eqref{eq:supdim}, investigations of algorithmic dimension focus almost entirely on the dimensions $\dim(A)$ of \emph{individual} languages (or, in other contexts, individual sequences or individual points in a metric space) $A$.

Turning to complexity classes, i.e., the cases where $\Delta$ is some resource bound $\up_i$ or $\up_i\uspace$, the dimension~\ref{defn:dim2}. is non-degenerate in the sense that $\dim(R(\Delta)\mid R(\Delta))=1$. If $X\subseteq R(\Delta)$ is finite or even ``$\Delta$-countable,'' then $\dim(X\mid R(\Delta))=0$. This implies for example that, for each fixed $k\in\N$,
\begin{equation}\label{eq:smallinbig}
	\dim(\TIME(2^{kn})\mid \E)=\dim(\TIME(2^{n^k})\mid \EXP)=0.
\end{equation}

Finally, we mention interactions of dimensions with randomness. A language $A\in\mathbf{C}$ is \emph{$\Delta$-random} if no $\Delta$-computable martingale succeeds on it~\cite{Lutz92}. A language $A\in\mathbf{C}$ is \emph{algorithmically random} (or \emph{Martin-L\"{o}f random}~\cite{MarL66}) if no lower semicomputable martingale succeeds on it. Since a martingale is a 1-gale, this implies that $\dim_\Delta(A)=1$ holds for every $\Delta$-random language and $\dim(A)=1$ holds for every algorithmically random language. In neither case does the converse hold.

\section{The Point-to-Set Principle}\label{sec:psp}

As noted in the introduction, previous instances of the Point-to-Set Principle have characterized \emph{classical} fractal dimensions of sets in terms of the relativized \emph{algorithmic} dimensions of the elements of these sets. Here we make the Point-to-Set Principle more widely applicable by proving instances of it in which ``classical'' and ``algorithmic'' are replaced by resource bounds $\Delta$ and $\Gamma$, respectively, with $\Gamma$ smaller (``more effective'') than $\Delta$.

To this end, we partially order our resource bounds by
\[\up_i<\up_{i+1}<\comp,\]
\[\up_i\uspace<\up_{i+1}\uspace<\comp,\]
and
\[\up_i\leq\up_i\uspace\]
for all $i\leq 1$ and
\[\comp<\all.\]
Aside from reflecting current knowledge about the inclusions among these classes, this ordering has the crucial property that, if $\Gamma$ and $\Delta$ are resource bounds with $\Gamma<\Delta$, then $\Delta$ \emph{parameterizes} $\Gamma$ in the sense that there is a function $f\in\Delta$ such that
\[\Gamma=\{f_k\mid k\in\N\},\]
where each $f_k:\{0,1\}^*\to\{0,1\}^*$ is the $k$\textsuperscript{th} \emph{slice} of $f$, defined by $f_k(x)=f(0^k1x)$ for all $x\in\{0,1\}^*$. Moreover, this parameterization \emph{relativizes} in the sense that, for each function oracle $g:\{0,1\}^*\to\{0,1\}^*$, there is a function $f^g\in\Delta^g$ such that
\[\Gamma^g=\{f_k^g\mid k\in\N\}.\]

\begin{theorem}\label{thm:precursor}
	If $\Gamma$ and $\Delta$ are resource bounds with $\Gamma<\Delta$, then for each function oracle $g:\{0,1\}^*\to\{0,1\}^*$, there is a $\Delta^g$-computable function $d^g$ such that $\{d_k^g\mid k\in\N\}$ is the set of all martingales that are $\Gamma^g$-computable and satisfy $d_k^g(\lambda)\leq 1$.
\end{theorem}
\begin{proof}
	This is implicit in the proofs of the time and space hierarchy theorems~\cite{HarSte65,HarLewSte65} (minus the ``disagreement'' step of the diagonalizations), together with the well-known fact that these proofs relativize.
\end{proof}

The following theorem is the main result of this section.
\begin{theorem}[Point-to-Set Principle for Resource-Bounded Dimensions]\label{thm:psprbd}
	If $\Gamma$ and $\Delta$ are resource bounds with $\Gamma<\Delta$, then, for all $X\subseteq\mathbf{C}$,
	\begin{equation}\label{eq:psprbd}
		\dim_\Delta(X)=\adjustlimits\min_{g\in\Delta}\sup_{A\in X}\dim_\Gamma^g(A).
	\end{equation}
\end{theorem}

Theorem~\ref{thm:psprbd} follows immediately from the following two lemmas, which we prove separately.

\begin{lem}\label{lem:pspub}
	If $\Gamma$, $\Delta$, and $X$ are as in Theorem~\ref{thm:psprbd} and $g\in\Delta$, then
	\begin{equation}\label{eq:pspub}
		\dim_\Delta(X)\leq\sup_{A\in X}\dim_\Gamma^g(A).
	\end{equation}
\end{lem}

\begin{lem}\label{lem:psplb}
	If $\Gamma$, $\Delta$, and $X$ are as in Theorem~\ref{thm:psprbd}, then there exists $g\in\Delta$ such that, for all $A\in X$,
	\begin{equation}\label{eq:psplb}
		\dim_\Gamma^g(A)\leq\dim_\Delta(X).
	\end{equation}
\end{lem}

\begin{proof}[Proof of Lemma~\ref{lem:pspub}]
	Let $\Gamma$, $\Delta$, $X$, and $g$ be as given, and let $s\in\Q$ satisfying
	\begin{equation}\label{eq:pspubpf1}
		s>\sup_{A\in X}\dim_\Gamma^g(A).
	\end{equation}
	It suffices to show that
	\begin{equation}\label{eq:pspubpf2}
		\dim_\Delta(X)\leq s.
	\end{equation}
	
	Since $\Gamma<\Delta$, Theorem~\ref{thm:precursor} tells us that there is a $\Delta^g$-computable function $d^g:\{0,1\}^*\to[0,\infty)$ such that the set $\{d_k^g\mid k\in\N\}$ of all slices of $d^g$ is the set of all martingales that are $\Gamma^g$-computable and satisfy $d_k^g(\lambda)\leq 1$. In fact, since $g\in\Delta$, this function $d^g$ is $\Delta$-computable. Define the function $d^{g,s}:\{0,1\}^*\to[0,\infty)$ so that
	\[d^{g,s}(0^k1x)=2^{(s-1)|x|}d^g(0^k1x)\]
	holds for all $k\in\N$ and $x\in\{0,1\}^*$. Then $d^{g,s}$ is $\Delta$-computable, and Observation~\ref{obs:diss} tells us that $\{d_k^{g,s}\mid k\in\N\}$ is the set of all $\Gamma^g$-computable $s$-gales that satisfy $d_k^{g,s}(\lambda)\leq 1$. Define $d:\{0,1\}^*\to[0,\infty)$ by
	\begin{equation}\label{eq:ddef}
		d=\sum_{k=0}^\infty 2^{-k}d_k^{g,s}.
	\end{equation}
	Then $d$ is a $\Delta$-computable $s$-gale, so to confirm~\eqref{eq:pspubpf2} it suffices to show that
	\begin{equation}\label{eq:pspubpf3}
		X\subseteq S^\infty[d].
	\end{equation}
	For this, let $A\in X$. Then, by~\eqref{eq:pspubpf1}, there is a $\Gamma^g$-computable $s$-gale $\tilde{d}$ such that $A\in S^\infty[\tilde{d}]$. Then there exists $k\in\N$ such that $d_k^{g,s}=\tilde{d}$, whence $A\in S^\infty[d_k^{g,s}]$. But then~\eqref{eq:ddef} tells us that
	\[\limsup_{w\to A}d(w)\geq 2^{-k}\limsup_{w\to A} d_k^{g,s}(w)=\infty,\]
	whence~\eqref{eq:pspubpf3} holds.
\end{proof}

\begin{proof}[Proof of Lemma~\ref{lem:psplb}]
	Let $\Gamma$, $\Delta$, and $X$ be as given, and let $s\in\Q$ satisfy
	\begin{equation}\label{eq:psplbpf1}
		s>\dim_\Delta(X).
	\end{equation}
	If suffices to exhibit $g\in\Delta$ such that, for all $A\in X$,
	\begin{equation}\label{eq:psplbpf2}
		\dim_\Gamma^g(A)\leq s.
	\end{equation}

	By~\eqref{eq:psplbpf1}, there is a $\Delta$-computable $s$-gale $d$ such that
	\begin{equation}\label{eq:psplbpf3}
		X\subseteq S^\infty[d].
	\end{equation}
	Let $g=\hat{d}\in\Delta$ testify to the $\Delta$-computability of $d$ as defined in Section~\ref{sec:rb}. Then $d$ is a $\Gamma^g$-computable $s$-gale, and~\eqref{eq:psplbpf3} tells us that, for all $A\in X$, $A\in S^\infty[d]$, whence~\eqref{eq:psplbpf2} holds.
\end{proof}

This completes the proof of Theorem~\ref{thm:psprbd}. We now discuss some of its instances.

We first address a small technical issue regarding relativization. Instances of the Point-to-Set Principle are usually stated in terms of oracles in $\mathbf{C}$ rather than in terms of function oracles as in Theorem~\ref{thm:psprbd}. These are equivalent for such large resource bounds as $\all$ and $\comp$, but some care is required for smaller resource bounds. For example, the case $\Gamma=\up$, $\Delta=\qp$ of Theorem~\ref{thm:psprbd} says that, for all $X\subseteq C$,
\begin{equation}\label{eq:pqp1}
	\dim_\qp(X)=\adjustlimits\min_{g\in\qp}\sup_{A\in X}\dim_\up^g(A).
\end{equation}
On the right-hand side, we would like to replace ``$g\in\qp$'' by ``$B\in R(\qp)$,'' i.e., ``$B\in\EXP$.'' However, this would \emph{not} be equivalent to~\eqref{eq:pqp1} and would in fact be false. The issue is that simulating an oracle query in the course of a computation of $d^B(w)$, where $d$ is $\up$-computable and $B\in\EXP$, could take $2^{|w|^k}$ time, which is \emph{not} within the $\qp$ resource bound on the left-hand side of~\eqref{eq:pqp1}. We thus introduce the special notation $\dim_\up^{\langle B\rangle}(A)$ for the $\up$-dimension of $A$ relative to $B\in\mathbf{C}$, with the proviso that a relativized $s$-gale $d^{\langle B\rangle}$ upper bounding $\dim_\up^{\langle B\rangle}$ is, inf computing $d^{\langle B\rangle}(w)$, only allowed to submit queries of length $O(\log |w|)$ to the oracle $B$.

With the above proviso, the instance~\eqref{eq:pqp1} of Theorem~\ref{thm:psprbd} says that, for all $X\subseteq \mathbf{C}$,
\begin{equation}\label{eq:pqp2}
	\dim_\qp(X)=\adjustlimits\min_{B\in\EXP}\sup_{A\in X}\dim_\up^{\langle B\rangle}(A).
\end{equation}
This implies that, for all $X\subseteq \mathbf{C}$,
\begin{equation}\label{eq:xinexp}
	\dim(X\mid\EXP)=\adjustlimits\min_{B\in\EXP}\sup_{A\in X\cap\EXP}\dim_\up^{\langle B\rangle}(A).
\end{equation}

The Point-to-Set Principle for Hausdorff dimension~\cite{LutLut17}, stated in the context of $\mathbf{C}$, says that, for all $X\subseteq\mathbf{C}$,
\begin{equation}\label{eq:psph}
	\dimH(X)=\adjustlimits\min_{B\in\mathbf{C}}\sup_{A\in X}\dim^{B}(A),
\end{equation}
thus characterizing the classical Hausdorff dimension of $X$ in terms of the relativized algorithmic dimensions of its individual elements. Since $\dim_\all=\dimH$, Theorem~\ref{thm:psprbd} tells us, for example, that we also have, for all $X\subseteq C$,
\begin{equation}\label{eq:psphp}
	\dimH(X)=\adjustlimits\min_{B\in\mathbf{C}}\sup_{A\in X}\dim_\up^{B}(A).
\end{equation}
Note that we could use $\dim_\up^{\langle B\rangle}(A)$ on the right-hand side here, but it is unnecessary, because the resource bound $\all$ on the left-hand side of~\eqref{eq:psphp} is unrestricted.

\section{Selectivity}\label{sec:sel}

\begin{defn}[\cite{Selm79}]
For any resource bound $\Delta$, a language $A\subseteq\{0,1\}^*$ is \emph{$\Delta$-selective} if there is a \emph{selector} function $f\in \Delta$ such that, for all pairs $a,b\in\{0,1\}^*$, we have $f(\langle a,b\rangle)\in\{a,b\}$ and
\[a\in A\ \text{or}\ b\in A\implies f(\langle a,b\rangle)\in A,\]
where $\langle\cdot,\cdot\rangle:\{0,1\}^*\times\{0,1\}^*\to\{0,1\}^*$ is a standard pairing function.
\end{defn}

\begin{theorem}\label{thm:pseldim}
	If $A,B\in\mathbf{C}$ and $g:\{0,1\}^*\to\{0,1\}^*$ are such that $B$ is $\up^g$-selective and $A\leq_m^P B$, then $\dim_\up^g(A)=0$.
\end{theorem}
\begin{proof}
	Let $A$, $B$, and $g$ be as in the theorem statement. Let $f\in\up^g$ be a selector for $A$, let $h:\{0,1\}^*\to\{0,1\}^*$ be a $\leq_m^P$-reduction from $A$ to $B$, and let $s>0$. We will show that $\dim_\up^g(A)\leq s$ by constructing an $s$-gale that succeeds on $A$ and is computable in polynomial time relative to $g$.

	Let $k\in\N$ be sufficiently large so that
	\begin{equation}\label{eq:choiceofk}
	  \frac{2^{ks}}{k+1}>1.
	\end{equation}
	We will consider blocks of $k$ consecutive strings. For each $q\in\N$, define the directed graph $G_q$ whose vertex set is $\{0,\ldots,k-1\}$ and edge set is
	\[\left\{(i,j)\mid f(\langle h(s_{q k+i}),h(s_{q k+j})\rangle)=h(s_{q k+j})\right\}.\]
	Notice that if $s_{q k+i}\in A$ and $s_{q k+j}\not\in A$, then $h(s_{q k+i})\in B$ and $h(s_{q k+j})\not\in B$. In this situation, the edge $(i,j)$ cannot be present in $G_q$, and more generally there cannot be any path from $i$ to $j$ in $G_q$.

	Let $G'_q$ be the directed acyclic graph obtained by contracting each strongly connected component of $G_q$ to a single vertex. Define a linear order $\prec_q$ on $\{0,\ldots,k-1\}$ by topologically sorting $G'_q$, breaking ties within each strongly connected component arbitrarily. In this order, $i\preceq_q j$ implies that there is a path from $i$ to $j$ in $G_q$.

	Thus, if $i\preceq_q j$ and $s_{q k+i}\in A$, then $s_{q k+j}\in A$. Extending $\prec_q$ by defining $i\prec_q k$ for all $i\in\{0,\ldots,k-1\}$, it follows that
	\begin{equation}\label{eq:threshold}
		A\cap\{s_{q k},\ldots,s_{q k+k-1}\}=\{s_{q k+j}\mid i\preceq_q j\}
	\end{equation}
	for some $i\in\{0,\ldots,k\}$.

	Define $d:\{0,1\}^*\to[0,\infty)$ and, for each $i\in\{0,\ldots,k\}$, $d_i:\{0,1\}^*\to[0,\infty)$ recursively as follows. For $w\in\{0,1\}^*$, let $qk+j=|w|$, where $j\in\{0,\ldots,k-1\}$.
	\begin{itemize}
		\item For all $i\in\{0,\ldots,k\}$, $d_i(\lambda)=d(\lambda)=1$.
		\item $d(w)=\frac{1}{k+1}\sum_{i=0}^k d_i(w)$.
		\item For all $i\in\{0,\ldots,k\}$ and $j=0$,
		\begin{align*}
			d_i(w0)&=
			\begin{cases}
				0&\text{if}\ i\preceq_q j\\
				2^{s}d(w)&\text{otherwise},
			\end{cases}\\
			d_i(w1)&=
			\begin{cases}
				2^{s}d(w)&\text{if}\ i\preceq_q j\\
				0&\text{otherwise},
			\end{cases}
		\end{align*}
		\item For all $i\in\{0,\ldots,k\}$ and $j\in \{1,\ldots,k-1\}$,
		\begin{align*}
			d_i(w0)&=
			\begin{cases}
				0&\text{if}\ i\preceq_q j\\
				2^{s}d_i(w)&\text{otherwise},
			\end{cases}\\
			d_i(w1)&=
			\begin{cases}
				2^{s}d_i(w)&\text{if}\ i\preceq_q j\\
				0&\text{otherwise}.
			\end{cases}
		\end{align*}
	\end{itemize}
	Informally, each $d_i$ represents a betting strategy, and $d$ is an aggregate betting strategy that evenly re-allocates between the $d_i$ after each block of $k$ bits. Observe that $d$ is an $s$-gale, although the individual $d_i$ are not.

	Now consider $d(A\upharpoonright n)$. If $n=0$, then $d(A\upharpoonright n)=1$. Otherwise, $n=qk+j$ for some $q\in\N$ and $j\in\{1,\ldots,k\}$. Let $i\in\{0,\ldots,k\}$ be the value satisfying equation~\eqref{eq:threshold} for this $q$. Then
	\begin{align*}
		d(X\upharpoonright n)&\geq \frac{d_i(X\upharpoonright n)}{k+1}\\
		&=\frac{2^{js}d(X\upharpoonright n-j)}{k+1}\\
		&=\frac{2^{(qk+j)s}}{(k+1)^q}\\
		&>\left(\frac{2^{ks}}{k+1}\right)^q.
	\end{align*}
	By inequality~\eqref{eq:choiceofk}, this lower bound is monotonically increasing and unbounded, so \[\liminf_{n\to\infty}d(A\upharpoonright n)=\infty.\] Therefore the $s$-gale $d$ succeeds on $A$. Furthermore, for all $w\in\{0,1\}^*$, the value $d(w)$ can be computed in polynomial time relative to $g$ by:
	\begin{itemize}
		\item $k$ calls to the polynomial-time reduction function $h$ on inputs \[s_{qk},\ldots,s_{qk+k-1},\]
		each of which has length $O(\log |w|)$;
		\item $k^2$ calls, for each ordered pairs from $\{s_{qk},\ldots,s_{qk+k-1}\}$, to the selector function $f$, which runs in polynomial time relative to $g$; and
		\item standard graph algorithms on $G_q$, which has $k=O(1)$ vertices.
	\end{itemize}
	We conclude that $\dim_\up^g(A)<s$, and the theorem follows immediately. 
\end{proof}

\begin{lem}\label{lem:h}
    Let $\qp'$ be the set of all functions in $\qp$ whose output length is polynomially bounded. There is a function $h\in \qp'$ such that $\qp'=\up^h$.
\end{lem}
\begin{proof}
    By standard techniques of clocking Turing machines and bounding their running times and output lengths, we can form an enumeration $M_0,M_1,M_2,\ldots$ of Turing machines such that $\qp'$ is exactly the set of functions computed by Turing machines in this list. Define $h:\{0,1\}^*\to\{0,1\}^*$ by
    \[
        h(u)=
        \begin{cases}
            M_k(x)&\text{if}\ u=0^k1x\\
            \lambda&\text{if $u$ does not contain a 1.}
        \end{cases}
    \]
    It is clear that $\up^h=\qp'$.
\end{proof}

\begin{theorem}
    $\dim(P_m(\qpSEL)\mid\EXP)=0$.
\end{theorem}
\begin{proof}
    Let $h$ be as in Lemma~\ref{lem:h}, and let $A\in P_m(\qpSEL)$. Then there exists some language $B\in\mathbf{C}$ and function $g\in \qp'=\up^h$ such that $A\leq_m^P B$ and $g$ is a selector for $B$, i.e., $B$ is $\up^h$-selective. By Theorem~\ref{thm:pseldim}, then, $\dim_\up^h(A)=0$. This holds for all $A\in P_m(\qpSEL)$, so we can apply Theorem~\ref{thm:psprbd}:
    \begin{align*}
    \dim_\qp(P_m(\qpSEL))
        &\leq\sup_{A\in P_m(\qpSEL)}\dim_\up^h(A)\\
        &=0.
    \end{align*}
    Since $\dim(P_m(\qpSEL)\mid\EXP)$ is defined as \[\dim_\qp(P_m(\qpSEL)\cap\EXP)\leq\dim_\qp(P_m(\qpSEL)),\]
    this completes the proof.
    
\end{proof}

\begin{cor}
    No $\qp$-selective set is partially $\leq_m^P$-hard for $\EXP$.
\end{cor}

\begin{cor}
	If $\dim(\NP\mid \EXP)>0$, then no $\qp$-selective set is $\leq^P_m$-hard for $\NP$.
\end{cor}

\section{Disjoint $\NP$ Pairs}\label{sec:disj}
In this section we improve the results in \cite{FoLuMa12}\ by proving that the dimension of $\disjNP$ in $\disjEXP$ is related to the dimension of NP inside EXP.

\begin{defn}[\cite{DSRSSSF,Lutz:DFRD}]
  For $s\in[0,\infty)$ and distribution $\beta$ on alphabet $\Sigma$, a \emph{$\beta$-$s$-gale} is a function $d:\Sigma^*\to[0,\infty)$ such that, for all $w\in\Sigma^*$,
  \[d(w)=\sum_{a\in\Sigma}d(wa)\beta(a)^s.\]
  A $\beta$-$s$-gale succeeds on a language $A\subseteq\Sigma^*$, and we write $A\in S^\infty[d]$, if
  \[\limsup_{w\to A}d(w)=\infty.\]
  Let $\Delta$ be a resource bound, $\beta$ a distribution on alphabet $\Sigma$, and $X\subseteq\mathcal{P}(\Sigma^*)$. Then $\mathcal{G}_{\Delta,\beta}(X)$ denotes the set of all $s\in[0,\infty)$ such that there is a $\Delta$-computable $\beta$-$s$-gale $d$ for which $X\subseteq S^\infty[d]$, and the \emph{$\Delta$-$\beta$-dimension} of $X$ is
  \[\dim_{\Delta,\beta}(X)=\inf\mathcal{G}_{\Delta,\beta}(X).\]
\end{defn}

We code disjoint pairs as in \cite{FoLuMa12}, using the alphabet \trialph.
For a pair $(A, B)$, $1$ corresponds to $A$, $-1$ to $B$, and $0$ to
$(A\cup B)^c$.

We fix a probability distribution $\gamma_0$ on \trialph\ as
$\gamma_0(0)=1/4$, $\gamma_0(1)=\gamma_0(-1)=3/8$, that is the
natural distribution used in \cite{FoLuMa12}. For disjoint pairs we write $\dimdelta(X)$ for $\dim_{\Delta,\gamma_0}(X)$. Theorem~\ref{thm:psprbd} extends routinely to this setting.

The main theorem of this section is the following

\begin{theorem}\label{the1}

If $\dim(\disjNP\mid\disjEXP)=1$, then $\dim(\NP\mid\EXP)>0$.

\end{theorem}

The proof of Theorem \ref{the1}\ is based on the following two
results and Theorem~\ref{thm:psprbd}.

\begin{theorem}\label{the2} Let $\beta$ be a positive
distribution on \bin, $X\subseteq\mathbf{C}$, and $g:\{0,1\}^*\to\{0,1\}^*$.
If $\dim_{\up}^g(X)=0$, then $\dim_{\up,\beta}^g(X)<1$.
\end{theorem}

\begin{theorem}\label{the3}
Let $\beta = (1/4, 3/4)$ and $g:\{0,1\}^*\to\{0,1\}^*$. If $\dim_{\up,\beta}^g(\NP) < 1$, then
\[\dim_{\up}^g(\disjNP)<1.\]
\end{theorem}

Theorem \ref{the2} is a consequence of the following lemma.

\begin{lem}\label{lem4}

Let $g:\{0,1\}^*\to\{0,1\}^*$, let $s$ be such that $\dim_{\up}^g(X)<s$, and let $\beta$ be a distribution
on \bin. If $\max(\beta(0), \beta(1))< 2^{-s}$, then
$\dim_{\up,\beta}^g(X)<1$.
\end{lem}

\begin{proof}[Proof of Lemma \ref{lem4}]
Let $s'> s$ and $t\in (0,1)$ be such that $\max(\beta(0), \beta(1))<2^{-s'/t}$. Let $d$
be a $\up^g$-computable $s$-gale. Define
\[d'(wb)=d'(w) \frac{d(wb)}{2^{s}d(w)}\frac{1}{\beta(b)^{t}}.\]

Then $d'$ is a $\up^g$-computable $\beta$-$t$-gale. Furthermore,

\[d'(w)\ge d (w) 2^{-s|w|}\frac{1}{\beta(w)^{t}} > d (w) 2^{-s|w|}
2^{s'|w|},\]
and therefore $\regSS[d]\subseteq \regSS[d']$. 

\end{proof}

Theorem \ref{the3}\ is a consequence of the following lemma.

\begin{lem}\label{lem5}

Let $g:\{0,1\}^*\to\{0,1\}^*$, $\gamma$ a positive distribution on \trialph, $\beta$ a distribution
on \bin\ with $\beta(0)=\gamma(0)$, and $X\subseteq\mathbf{C}$ a class that is closed under union. If $\dim_{\up,\beta}^g(X)<1$,
then $\dim_{\up,\gamma}^g(\disj X)<1$.

\end{lem}

\begin{proof}[Proof of Lemma \ref{lem5}] If 
$\dim_{\up,\beta}^g(X)<s<1$ and $d$ is a $\up^g$-computable $\beta$-$s$ gale succeeding on $X$,
let
$ s' \in (0,1)$ with $\beta(1)^s \ge \gamma(1)^{s'}+
\gamma(-1)^{s'}$ and $\beta(0)^s\ge \gamma(0)^{s'}$.

We define a $\up^g$-computable $\gamma$-$s'$ gale $D$ by
\begin{eqnarray*}D(w0)&=&D(w)\frac{d(\overline{w}0)}{d(\overline{w})}\frac{\beta(0)^s}{\gamma(0)^{s'}}\\
D(w1)=D(w-1)&=&D(w)\frac{d(\overline{w}1)}{d(\overline{w})}\frac{\beta(1)^s}{\gamma(1)^{s'}+\gamma(-1)^{s'}},\\
\end{eqnarray*}
where
\begin{eqnarray*}\overline{w}[i]=0&\mbox{if}&w[i]=0\\
\overline{w}[i]=1&\mbox{if}&w[i]=1\ \mathrm{ or }\ w[i]=-1.\\
\end{eqnarray*}
That is, if $w$ is a prefix of $(A, B)$ then $\overline{w}$ is a
prefix of $A\cup B$.

Notice that $D(w)\ge d(\overline{w})$ for every $w$.

Thus if $(A,B)\in\disj X$, then $A\cup B \in X$ and $D$ succeeds on
$(A, B)$.
\end{proof}

\begin{proof}[Proof of Theorem \ref{the1}]
	We prove the contrapositive. Suppose that
	 $\dim(\NP\mid\EXP)=0$. By Theorem~\ref{thm:psprbd}, there is a $g\in\qp$ such that $\dim_{\up}^g(\NP)=0$.
	 
	 Let $\beta = (1/4, 3/4)$. By Theorem~\ref{the2}, $\dim_{\up,\beta}^g(\NP)<1$. By Theorem~\ref{the3},
	 $\dim_{\up}^g(\disjNP)<1$.
	 
	 Using Theorem~\ref{thm:psprbd} again, $\dim(\disjNP\mid\disjEXP)= \dim_{\qp}(\disjNP)<1$.
\end{proof}

\bibliographystyle{plain}
\bibliography{dscc}

\begin{thebibliography}{10}
\providecommand{\url}[1]{{#1}}
\providecommand{\urlprefix}{URL }
\expandafter\ifx\csname urlstyle\endcsname\relax
  \providecommand{\doi}[1]{DOI~\discretionary{}{}{}#1}\else
  \providecommand{\doi}{DOI~\discretionary{}{}{}\begingroup
  \urlstyle{rm}\Url}\fi

\bibitem{AMRS01}
Ambos{-}Spies, K., Merkle, W., Reimann, J., Stephan, F.: Hausdorff dimension in
  exponential time.
\newblock In: Proceedings of the 16th Annual {IEEE} Conference on Computational
  Complexity, Chicago, Illinois, USA, June 18--21, 2001, pp. 210--217. {IEEE}
  Computer Society (2001)

\bibitem{BisPer17}
Bishop, C.J., Peres, Y.: Fractals in Probability and Analysis.
\newblock Cambridge University Press (2017)

\bibitem{BuhLon96}
Buhrman, H., Longpr\'{e}, L.: Compressibility and resource bounded measure.
\newblock In: Proceedings of the Thirteenth Symposium on Theoretical Aspects of
  Computer Science, pp. 13--24. Springer-Verlag, Berlin (1996)

\bibitem{DosGla20}
Dose, T., Gla{\ss}er, C.: {NP}-completeness, proof systems, and disjoint
  {NP}-pairs.
\newblock In: C.~Paul, M.~Bl{\"{a}}ser (eds.) 37th International Symposium on
  Theoretical Aspects of Computer Science, {STACS} 2020, March 10--13, 2020,
  Montpellier, France, \emph{LIPIcs}, vol. 154, pp. 9:1--9:18. Schloss Dagstuhl
  - Leibniz-Zentrum f{\"{u}}r Informatik (2020)

\bibitem{EvSeYa84}
Even, S., Selman, A.L., Yacobi, Y.: The complexity of promise problems with
  applications to public-key cryptography.
\newblock Information and Control \textbf{61}(2), 159--173 (1984)

\bibitem{Falc14}
Falconer, K.: Fractal Geometry: Mathematical Foundations and Applications, 3rd
  edition.
\newblock John Wiley \& {S}ons (2014)

\bibitem{FoLuMa12}
Fortnow, L., Lutz, J.H., Mayordomo, E.: Inseparability and strong hypotheses
  for disjoint {NP} pairs.
\newblock Theory of Computing Systems \textbf{51}, 229--247 (2012)

\bibitem{GHSW14}
Gla{\ss}er, C., Hughes, A., Selman, A.L., Wisiol, N.: Disjoint {NP}-pairs and
  propositional proof systems.
\newblock {SIGACT} News \textbf{45}(4), 59--75 (2014)

\bibitem{GlSeSe05}
Gla{\ss{}}er, C., Selman, A.L., Sengupta, S.: Reductions between disjoint
  {NP}-pairs.
\newblock Information and Computation \textbf{200}, 247--267 (2005)

\bibitem{GSSZ04}
Gla{\ss{}}er, C., Selman, A.L., Sengupta, S., Zhang, L.: Disjoint {NP}-pairs.
\newblock SIAM Journal on Computing \textbf{33}, 1369--1416 (2004)

\bibitem{GlSeZh07}
Gla{\ss{}}er, C., Selman, A.L., Zhang, L.: Canonical disjoint {NP}-pairs of
  propositional proof systems.
\newblock Theoretical Computer Science \textbf{370}, 60--73 (2007)

\bibitem{GlSeZh09}
Gla{\ss{}}er, C., Selman, A.L., Zhang, L.: The informational content of
  canonical disjoint {NP}-pairs.
\newblock Int. J. Found. Comput. Sci. \textbf{20}(3), 501--522 (2009)

\bibitem{GroSel88}
Grollmann, J., Selman, A.: Complexity measures for public-key cryptosystems.
\newblock SIAM J. Comput. \textbf{11}, 309--335 (1988)

\bibitem{DSRSSSF}
Gu, X., Lutz, J., Mayordomo, E., Moser, P.: Dimension spectra of random
  subfractals of self-similar fractals.
\newblock Annals of Pure and Applied Logic \textbf{165}, 1707--1726 (2014)

\bibitem{HarSte65}
Hartmanis, J., Stearns, R.: On the computational complexity of algorithms.
\newblock Transactions of the American Mathematical Society \textbf{117},
  285--306 (1965)

\bibitem{Haus19}
Hausdorff, F.: Dimension und {\"a}u{\ss}eres {M}a{\ss}.
\newblock Math. Ann. \textbf{79}, 157--179 (1919)

\bibitem{HemTor03}
Hemaspaandra, L.A., Torenvliet, L.: Theory of Semi-Feasible Algorithms.
\newblock Springer-Verlag (2002)

\bibitem{HomSel92}
Homer, S., Selman, A.L.: Oracles for structural properties: The isomorphism
  problem and public-key cryptography.
\newblock Journal of Computer and System Sciences \textbf{44}, 287--301 (1992)

\bibitem{Jock68}
Jockusch, C.G.: Semirecursive sets and positive reducibility.
\newblock Trans. Amer. Math. Soc. \textbf{131}, 420--436 (1968)

\bibitem{LuLuMa20}
Lutz, J., Lutz, N., Mayordomo, E.: Extending the reach of the point-to-set
  principle (2020).
\newblock \urlprefix\url{https://arxiv.org/pdf/2004.07798.pdf}

\bibitem{LutzThesis}
Lutz, J.H.: Resource-bounded category and measure in exponential complexity
  classes.
\newblock Ph.D. thesis, California Institute of Technology (1987)

\bibitem{Lutz90}
Lutz, J.H.: Category and measure in complexity classes.
\newblock SIAM Journal on Computing \textbf{19}, 1100--1131 (1990)

\bibitem{AEHNC}
Lutz, J.H.: Almost everywhere high nonuniform complexity.
\newblock J.~Comput. Syst. Sci. \textbf{44}(2), 220--258 (1992)

\bibitem{Lutz92}
Lutz, J.H.: Almost everywhere high nonuniform complexity.
\newblock Journal of Computer and System Sciences \textbf{44}(2), 220--258
  (1992)

\bibitem{QSET}
Lutz, J.H.: The quantitative structure of exponential time.
\newblock In: L.A. Hemaspaandra, A.L. Selman (eds.) Complexity Theory
  Retrospective {II}, pp. 225--254. Springer-Verlag (1997)

\bibitem{Lutz03a}
Lutz, J.H.: Dimension in complexity classes.
\newblock {SIAM} J. Comput. \textbf{32}(5), 1236--1259 (2003)

\bibitem{DISS}
Lutz, J.H.: The dimensions of individual strings and sequences.
\newblock Information and Computation \textbf{187}(1), 49--79 (2003)

\bibitem{Lutz:DFRD}
Lutz, J.H.: A divergence formula for randomness and dimension.
\newblock Theoretical Computer Science \textbf{412}, 166--177 (2011)

\bibitem{LutLut17}
Lutz, J.H., Lutz, N.: Algorithmic information, plane {K}akeya sets, and
  conditional dimension.
\newblock In: Proceedings of the 34th Symposium on Theoretical Aspects of
  Computer Science, {STACS} 2017, March 8--11, 2017, Hannover, Germany, pp.
  53:1--53:13 (2017)

\bibitem{LutLut20}
Lutz, J.H., Lutz, N.: Who asked us? {H}ow the theory of computing answers
  questions about analysis.
\newblock In: D.~Du, J.~Wang (eds.) Complexity and Approximation: In Memory of
  Ker-I Ko, pp. 48--56. Springer (2020)

\bibitem{LutMay21}
Lutz, J.H., Mayordomo, E.: Algorithmic fractal dimensions in geometric measure
  theory.
\newblock In: V.~Brattka, P.~Hertling (eds.) Handbook of Computability and
  Complexity in Analysis, pp. 271--302. Springer (2021)

\bibitem{Lutz21}
Lutz, N.: Fractal intersections and products via algorithmic dimension.
\newblock ACM Trans. Comput. Theory \textbf{13}(3) (2021)

\bibitem{LutStu18}
Lutz, N., Stull, D.M.: Projection theorems using effective dimension.
\newblock In: 43rd International Symposium on Mathematical Foundations of
  Computer Science, {MFCS} 2018, August 27-31, 2018, Liverpool, {UK}, pp.
  71:1--71:15 (2018)

\bibitem{LutStu20}
Lutz, N., Stull, D.M.: Bounding the dimension of points on a line.
\newblock Information and Computation \textbf{275} (2020)

\bibitem{MarL66}
Martin-L{\"o}f, P.: The definition of random sequences.
\newblock Information and Control \textbf{9}, 602--619 (1966)

\bibitem{Matt15}
Mattila, P.: Fourier Analysis and Hausdorff Dimension.
\newblock Cambridge Studies in Advanced Mathematics. Cambridge University Press
  (2015)

\bibitem{Razb94}
Razborov, A.: On provably disjoint {NP} pairs.
\newblock Tech. Rep. 94-006, ECCC (1994)

\bibitem{Selm79}
Selman, A.: P-selective sets, tally languages, and the behavior of polynomial
  time reducibilities on {NP}.
\newblock Mathematical Systems Theory \textbf{13}, 55–65 (1979)

\bibitem{Selm89}
Selman, A.: Complexity issues in cryptography.
\newblock In: Computational complexity theory (Atlanta, GA, 1988), \emph{Proc.
  Sympos. Appl. Math.}, vol.~38, pp. 92--107. Amer. Math. Soc. (1989)

\bibitem{HarLewSte65}
Stearns, R.E., Hartmanis, J., Lewis, P.: Hierarchies of memory limited
  computations.
\newblock In: Proc. 6th Annual Symp. on Switching Circuit Theory and Logical
  Design, pp. 179--190 (1965)

\bibitem{SteSha05}
Stein, E.M., Shakarchi, R.: Real Analysis: Measure Theory, Integration, and
  Hilbert Spaces.
\newblock Princeton Lectures in Analysis. Princeton University Press (2005)

\bibitem{Wang96}
Wang, Y.: Randomness and complexity.
\newblock Ph.D. thesis, Department of Mathematics, University of Heidelberg
  (1996)

\bibitem{Zima04}
Zimand, M.: Computational Complexity: A Quantitative Perspective.
\newblock Elsevier, Amsterdam (2004)

\end{thebibliography}

\end{document}